\documentclass[letterpaper,11pt,reqno]{amsart}

\makeatletter
\usepackage{amssymb}
\usepackage{latexsym}
\usepackage{amsbsy}
\usepackage{amsfonts}
\usepackage{graphicx}
\usepackage{enumerate}
\usepackage{enumitem}
\usepackage{mathtools}
\usepackage{subcaption}
\usepackage{color}
\usepackage{tcolorbox}
\usepackage{url}
\usepackage{comment}
\usepackage{appendix}

\usepackage{tikz}

\def\marginpar#1{\ignorespaces}

  \newcommand{\beq}{\begin{equation}}
    \newcommand{\eeq}{\end{equation}}
    
    \newcommand{\bal}{\begin{align}}
    \newcommand{\eal}{\end{align}}
    \newcommand{\bals}{\begin{align*}}
    \newcommand{\eals}{\end{align*}}
    
\newtheorem{theorem}{Theorem}[section]

\newtheorem{lemma}[theorem]{Lemma}
\newtheorem{proposition}[theorem]{Proposition}
\newtheorem{corollary}[theorem]{Corollary}

\numberwithin{equation}{section}
\makeatother
\begin{document}
\title[PoS, CEX and MFG]{Proof of Stake economy under centralized exchanges -- a mean field model}

\author[W. Tang]{Wenpin Tang}

\address[W. Tang]{Department of Industrial Engineering and Operations Research, Columbia University, S.W. Mudd Building, 
500 W 120th St, New York, NY 10027} 
\email{wt2319@columbia.edu}

\date{\today}

\begin{abstract}
We consider the interaction between centralized trading and decentralized Proof of Stake (PoS) blockchain ecosystems. 
Motivated by the increasing dominance of centralized exchanges and the institutionalization of crypto markets,
we study how trading activities on centralized exchanges affect staking behavior, token allocation, and decentralization within a PoS blockchain.
We formulate a continuous-time mean field model,
where the miners simultaneously act as validators in the PoS protocol and traders in a centralized market with price impact.
Under suitable assumptions, we establish  the local well-posedness of the mean field system,
and derive a semi-explicit characterization of the equilibrium trading strategy. 
Numerical results suggest that centralized trading activities (net pruchases) enhance staking participation, and promote decentralization of the staking distribution through market incentives
-- this partially justifies \texttt{BitMine}'s buy-and-stake strategy for ETH.
We also study the effects of transaction costs and token supply mechanisms on the equilibrium staking ratio and concentration profile. 
These results illustrate how market microstructure and centralized liquidity provision can exert significant influence on decentralized blockchain protocols.
\end{abstract}

\maketitle



\section{Introduction}
\label{sc1}

\quad A blockchain is a distributed ledger allowing the secure transfer of assets in a network without an intermediary,
hence achieving decentralization. 
In the past decade, 
blockchain technology has advanced tremendously,
with a range of applications including 
cryptocurrencies \cite{Naka08, Wood14},
healthcare \cite{MC19, TPE20},
supply chain \cite{CTT20, Kam18},
and prediction markets \cite{Augur, Poly}. 
At the core of blockchain technology is the consensus protocol, 
which specifies a set of voting rules
for the miners or validators to agree on an ever-growing log of transactions
so as to form a distributed ledger.
In early days, 
Proof of Work (PoW), which relies on computing resources, is used as the consensus protocol.
Due to high energy consumption, there is a strong incentive among blockchain practitioners
to switch from PoW to Proof of Stake (PoS) by committing stakes for reaching consensus in blockchains.
An example is Ethereum's Merge in September 2022 \cite{Merge}.

\quad One major consequence of blockchain technology is the emergence of decentralized finance (DeFi),
which was initially driven by the development of automated market makers (AMMs) known as Uniswap \cite{Uni2, Uni3}.
The idea of DeFi is to replace traditional financial intermediaries in centralized exchanges (CEXs) with decentralized exchanges (DEXs) governed by smart contracts. 
In contrast to CEXs, where liquidity provision and price discovery are controlled by centralized entities, AMMs enable peer-to-peer trading through algorithmic liquidity pools deployed on the blockchain,
thereby improving transparency, accessibility, and composability of financial services. 
See e.g., \cite{CJ21, MMT24, MMTZ22} for arbitrage mechanisms and pricing in AMMs,
and \cite{CCM26, CZ25, CDM25, MC25} for optimal execution and trading in AMMs.

\quad Despite the success of AMMs and DEXs, many modern DeFi protocols increasingly incorporate hybrid mechanisms combining AMM-based liquidity with order-book-like structures in an attempt to achieve higher trading efficiency and lower transaction costs.
See Figure \ref{fig:0} for an illustration of DeFi evolution.
CEXs such as Coinbase and Binance still dominate the cryptocurrency market in terms of trading volume, 
accounting for roughly $80\%-90\%$ of total spot trading activity \cite{Lim}.
Moreover, a significant portion of crypto exposure now flows through regulated financial products such as crypto ETFs, further strengthening the role of centralized financial infrastructure.
At the regulatory level, recent developments such as the Clarity Act \cite{Clarity}  indicate the growing institutionalization of digital asset markets, potentially accelerating the convergence between DeFi and traditional centralized financial systems.
\begin{figure}[!htbp]
    \centering
    \includegraphics[width=0.9\textwidth]{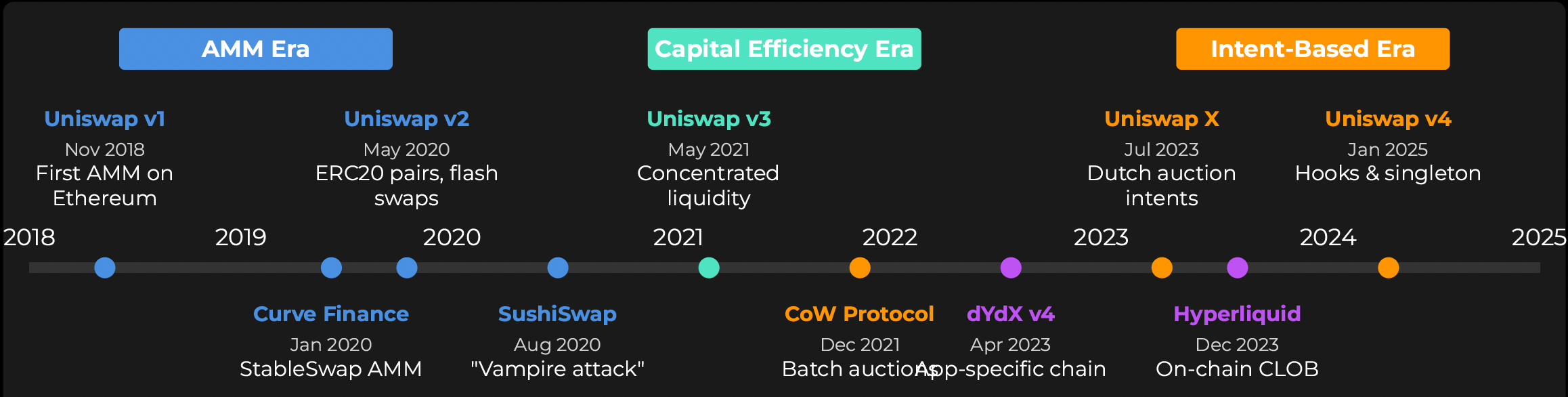}
     \caption{DeFi evolution (courtesy of Ruizhe Jia).}
    \label{fig:0}
\end{figure}

\quad The purpose of this paper is to study the impact of centralized trading activities on decentralized blockchain ecosystem,
which differs from the aforementioned AMM-based literature.
Here we consider the PoS protocol,
where there is a bidding mechanism to select a miner to do the work of validating a new block.
Each miner is required to commit some tokens, 
and the winning probability is proportional to the number of tokens committed.
The selected miner will create a new block, 
and get some crypto tokens as a reward\footnote{The reward includes that issued by the PoS ecosystem, and transaction fees. See e.g., \url{https://ethereum.org/staking/} for Ethereum staking rewards.}.
With the presence of trading exchanges, 
a miner can play a dual role both as a validator of the PoS blockchain
and as a trader/investor of the PoS tokens. 
Here we do not plan to address the security issues of this dual role, which is a substantial problem on its own.
Rather, our focus is on the economic side of the miner's dual role.
For a miner, there are two revenue resources:
(1) staking reward from the PoS ecosystem; 
(2) profit or loss from trading PoS tokens.
The question is how each miner makes decision to maximize her profit by exploiting or trading off the two revenue contributions,
and how the collective behavior of the miners affects the PoS ecosystem in turn.

\quad In this paper, we consider the PoS ecosystem in the presence of a centralized exchange,
following the approach outlined in \cite[Section 5]{Tangsuv}.
In this setting, we can define a notion of {\em equilibrium trading strategy} for a representative miner,
taking into account the interactions and competitive behavior among the miners and other investors on the exchange.
This leads to a mean field problem, 
which we provide a semi-closed form solution via a fixed-point equation.
As a by-product, we analyze the evolution of the staking profile of the miner population from a mean field perspective.
By solving numerically the fixed-point equation,
we also observe that under (centralized) trading,
\begin{itemize}[itemsep = 3 pt]
\item
the staking ratio increases over time;
\item
the staking profile decentralizes over time;
\item
the staking ratio increases slowly as the transaction costs become aggressive;
\item
the staking profile is more decentralized as the token reward grows faster,
\end{itemize}
demonstrating the market power over blockchain protocols.
We also mention the work \cite{BBLL20, LRS19, PW21} on the mean field game formulation of the PoW protocol,
which does not involve market microstructure.

\medskip
{\bf Organization of the paper}: The remainder of the paper is organized as follows. 
Section \ref{sc2} provides background and the mean field problem formulation.
In Section \ref{sc3}, we prove the well-posedness of the proposed mean field problem.
Numerical experiments are reported in Section \ref{sc4}.
Concluding remarks are summarized in Section \ref{sc5}.

\medskip
{\bf Notations}: Below we collect a few notations that will be used throughout.
\begin{itemize}[itemsep = 3 pt]
\item
$\mathbb{R}$ denotes the set of real numbers. 
\item
For a function $f: [0,\infty) \times \mathbb{R} \to \mathbb{R}$, 
$\partial_t f$ and $\partial_x f$ are the time- and space-derivatives of $f$.
\item
For a function $f: \mathcal{X} \to \mathbb{R}$ on a set $\mathcal{X}$,
$||f||_\infty: = \sup_{x \in \mathcal{X}}|f(x)|$ is the sup-norm of $f$.
\item
For a random variable $X$, $\mathbb{E}X$ denotes the expectation of $X$.
\end{itemize}

\section{Problem formulation}
\label{sc2}

\quad In this section, we introduce the mean field problem for trading under the PoS protocol.
We follow the presentation in \cite[Section 5]{Tangsuv},
with a few modifications that we will make precise.

\subsection{The PoS model}
Time is continuous, indexed by $t \in [0,T]$, 
for a fixed $T > 0$ representing the length of a finite horizon.
Let $\{N(t), \, 0 \le t \le T\}$ (with $N(0):= N$) denote the process of the supply of tokens,
which is increasing in time and sufficiently smooth.
So the derivative $N'(t)$ represents the instantaneous inflation rate, 
or rate of ``reward''  by the PoS protocol.
This rate is dynamically adjusted by blockchain community consensus, 
which includes the issuance of new tokens to miners and the burning of transaction fees 
(see \cite{ETHfound} for ETH supply mechanism). 
Its value depends on the overall network usage and adoption, and is typically stochastic.
Here we decouple the token supply with miner's (strategic) trading,
so the randomness in token supply is considered to be {\em external} in our setting.
Without loss of generality, we assume that $\{N(t), \, 0 \le t \le T\}$ is a {\em deterministic} process\footnote{For instance, the process $N(t)$ can take a polynomial form: 
\begin{equation*}
N_\alpha(t) = (N^{\frac{1}{\alpha}} + ct)^\alpha, \qquad t \ge 0.
\end{equation*}
This parametric family covers different token supply schemes according to the values of $\alpha$:
for $0< \alpha < 1$, the process $N_{\alpha}(t)$ corresponds to slowed inflation;
for $\alpha = 1$, the process $N_{1}(t) = N + ct$ gives a constant inflation rate $c$;
for $\alpha > 1$, the process $N_{\alpha}(t)$ amounts to accelerated inflation.
See \cite{Lev26} for discussions on ETH supply, which suggests slowed inflation.}.

\quad Let $K \ge 2$ be the number of miners/stakers, who are indexed by $k \in [K]:= \{1, \ldots, K\}$.
For each miner $k$, 
$\{X_k(t), \, 0 \le t \le T\}$ (with $X_k(0) = x_k$) denotes the number of tokens that miner $k$ chooses to stake for PoS mining,
and $\{\nu_k(t), \, 0 \le t \le T\}$ denotes miner $k$'s trading strategy.
Let $\{Z(t), \, t \ge 0\}$ (with $Z(0) = z$) be the process of the number of tokens that are not staked for the PoS mining\footnote{$Z(t)$ may have contributions from both the miners and retail investors. 
For the miners, they decide to stake a number of tokens for PoS mining, and set aside the remaining for other opportunities such as trading and DeFi lending/yield farming (e.g., Aave and Yearn Finance). 
For retail investors, they can also choose to stake some of their tokens. 
In this case, they will delegate their tokens to a PoS miner $k$, which is included in $X_k(t)$.},
so
\begin{equation}
\label{eq:conserv}
\sum_{k = 1}^K X_k(t) + Z(t) = N(t),  \quad \mbox{for } 0 \le t \le T.
\end{equation}
That is, only $N(t) - Z(t)$ tokens are locked for the PoS mining.

\quad In the (discrete) PoS protocol, in each round of the bidding process,
individual miners commit tokens so as to be elected to validate the block and receive a reward;
and the winning probability is proportional to the number of tokens committed. 
For instance, each round in Ethereum takes about $10$ seconds, corresponding to the the block-generation time \cite{BV14}.
In the continuous-time setting, 
time required for each round of voting is ``infinitesimal''.
Following \cite{TY23b}, the dynamics of miner $k$'s tokens under trading is formulated as\footnote{We assume that newly issued tokens are automatically staked.
In \cite{TY23b}, it is assumed that all tokens are staked for PoS mining.
Thus, $Z(t) \equiv 0$ and the dynamics \eqref{eq:Xnu2} reduces to $X'_k(t) = \nu_{k}(t) + \frac{N'(t)}{N(t)} X_k(t)$ for $0 \le t \le  \mathcal{T}_k$.
For ETH, it is observed that only around $30\%$ of total supply is locked in staking \cite{Lev26}. 
This necessities introducing the process of unstaked tokens $Z(t)$.}:
\begin{equation}
\label{eq:Xnu2}
X'_k(t) = \nu_{k}(t) + \frac{N'(t)}{N(t) - Z(t)} X_k(t) \quad \mbox{for } 0 \le t \le \tau_k  \wedge T:= \mathcal{T}_k,
\end{equation}
where $\tau_k: = \inf\{t>0: X_k(t) = 0\}$.
We set $X_k(t) = X_k(\mathcal{T}_k)$ for $t > \mathcal{T}_k$.
The equations \eqref{eq:conserv}--\eqref{eq:Xnu2} together imply that
\begin{equation*}
\sum_{k = 1}^K \nu_k(t) + Z'(t) = 0,
\end{equation*}
which can be viewed as a clearing house condition.
The instantaneous staking rate is $N'(t) - Z'(t)$,
where $N'(t)$ is from staking insurance
and $-Z'(t)$ is from trading between unstaked and staked tokens.
If $-Z'(t) >0$, the miners buy unstaked tokens for staking;
if $-Z'(t) < 0$, staked tokens are unstaked for trading.
As a result,
\begin{equation}
\label{eq:Z}
Z(t) = Z(0) - \int_0^t \sum_{k = 1}^K \nu_k(s) ds.
\end{equation}

\subsection{Price formation}
Central to each miner's decision is the price process $\{P(t), \, t \ge 0\}$ of each (unit) of token. 
As mentioned in the introduction, 
crypto trading remains centralized, with CEXs handling most volume via limit order books.
Our goal is to understand the PoS ecosystem under centralized trading mechanism.
Modern market microstructure theory postulates that market impact is the underlying driver of price dynamics,
i.e., trading volume impacts price movement.
For ease of presentation,
we assume linear price impact\footnote{Here we neglect market shocks triggered by macroeconomics,
geopolitics, breaking news, etc.}, i.e., the Almgren-Chriss model \cite{AC01}:
\begin{equation}
\label{eq:ACP}
P(t) = P(0) + \sigma B(t) - \eta(Z(t) - Z(0)),
\end{equation}
where $\{B(t), \, t \ge 0\}$ is standard Brownian motion,
$\sigma >0$ is the volatility, and $\eta > 0$ is a market impact parameter (known as Kyle's Lambda \cite{Kyle85}).
Refer to \cite[Chapter 3]{Gueant16} for background,
and \cite{DB15, GS11, TL11} for other market impact models.

\subsection{Consumption-investment problem}
We formulate the consumption-investment problem for each miner, 
following the idea in \cite{RS21, TY23b}.

\quad Let $b_{k}(t)$ be the (units of) risk-free asset that miner $k$ holds at time $t$,
and let $r > 0$ be the risk-free rate.
Assume that all participants (including the $K$ miners and retail investors) are allowed to trade tokens on a CEX,
whereas each miner can only exchange cash with an external source (say, a bank).
Let $\{c_k(t), \, 0 \le t \le T\}$ be the consumption process of miner $k$,
which evolves as:
\begin{equation}
dc_k(t) = rb_k(t) dt -db_k(t) - P(t) \nu_k(t) dt - N'(t) \, L\left(\frac{\nu_k(t)}{N'(t)}\right) dt,  \qquad 
0 \le t \le \mathcal{T}_k, \tag{C1}
\end{equation}
where $L(\cdot)$ is an even function, increasing on $\mathbb{R}_+$, strictly convex and asymptotically super-linear.
(C1) is a self-financing condition, 
where $rb_k(t) dt -db_k(t)$ is the net change in cash used to 
finance new tokens $P(t) \nu_k(t) dt$, transaction costs $N'(t) \, L\left(\frac{\nu_k(t)}{N'(t)}\right) dt$,
and consumption $dc_k(t)$.
Compared to \cite{TY23b}, there is an additional term $- N'(t) \, L\left(\frac{\nu_k(t)}{N'(t)}\right)$
representing for transaction costs:
it depends not only on the traded volume $\nu_k(t) dt$ 
but also the total volume $N'(t) dt$ (see \cite[p.43, (3.3)]{Gueant16}).
The quadratic cost $L(x) = \rho |x|^2$ with $\rho > 0$ corresponds to the {\em Almgren-Chriss model}.
We also impose the no-shorting constraint:
 \begin{equation}
b_k(0)  = 0, \quad b_k(t) \ge 0 \mbox{ for } 0 \le t \le \mathcal{T}_k, \quad 0 \le X_k(t) \le N(t) \mbox{ for } 0 \le t \le \mathcal{T}_k. \tag{C2}
 \end{equation}
Set $b_k(t) = b_k(\mathcal{T}_k)$ and $\nu_k(t) = 0$ for $t > \mathcal{T}_k$.
We note that the dynamics of DeFi markets, e.g., AMM, 
exhibit fundamentally different behaviors.\footnote{
For a constant-product AMM, let $R(t)$ be its token reserve and $Y(t)$ be its cash reserve
such that $R(t) Y(t) = k$. 
If a miner trade tokens from the AMM at rate $\nu_k(t)$, 
then $R'(t) = -\nu_k(t)$.
The AMM price is:
\begin{equation*}
P(t) = \frac{Y(t)}{X(t)} = \frac{k}{\left( R(0) - \int_0^t \nu_k(s)ds\right)^2}.
\end{equation*}
So the price is nonlinear in cumulative order flow $\int_0^t \nu_k(s)ds$,
which is different from the linear price impact \eqref{eq:ACP}.
If $\int_0^t \nu_k(s)ds \ll R_0$, Taylor expansion gives 
$P_t \approx P_0 + \frac{2P_0}{R_0} \int_0^t \nu_k(s)ds$.
In addition, over an infinitesimal interval $dt$, the cost is:
\begin{equation*}
\frac{k}{R(t) - \nu_k(t) dt} - \frac{k}{R(t)} = P(t) \nu_k(t)dt + \frac{P(t)}{R(t)} \nu_k(t)^2 dt^2 + \cdots.
\end{equation*}
Thus, to first order $dt$, the cash outflow is simply $P(t) \nu_k(t)dt$,
and does not incur an extra term $N'(t) \, L\left(\frac{\nu_k(t)}{N'(t)}\right) dt$ as in (C1).}

\smallskip
{\bf Miner's strategy if $Z(\cdot)$ is known}.
It is easy to see from \eqref{eq:ACP} that 
the price $P(t)$ (up to noise) only depends on the number of unstaked tokens $Z(t)$,
or equivalently, the number of all staked tokens $\sum_{k = 1}^K X_k(t)$.

\quad Here, suppose that each miner $k$ ``knows" the number of unstaked tokens.
The objective of miner $k$ is\footnote{In \cite{Tangsuv, TY23b}, the objective of each miner $k$ is set to be
\begin{equation*}
\begin{aligned}
\sup_{\{(\nu_k(t), b_k(t))\}} & J(\nu_k, b_k):=
\mathbb{E}\left\{ \int_0^{\mathcal{T}_k}e^{-\beta_k t} \left[dc_k(t) + \ell_k(X_k(t)) dt \right]  + e^{-\beta_k \mathcal{T}_k}  \left[b_k(\mathcal{T}_k) + h_k(X_k(\mathcal{T}_k) \right] \right\}  \\
& \mbox{ subject to } \eqref{eq:Xnu2}, \eqref{eq:ACP}, (\mbox{C}1), (\mbox{C}2),
\end{aligned}
\end{equation*}
for some utility functions $\ell_k(\cdot)$ and $h_k(\cdot)$.
The utilities are expressed as functions of the
{\em number of stakes}, as opposed to their {\em total value}. 
In this paper, we take market impact into consideration;
so it is more reasonable to include the total value in the objective.}:
\begin{equation}
\label{eq:obj11}
\begin{aligned}
\sup_{\{(\nu_k(t), b_k(t))\}} & J(\nu_k, b_k):=
\mathbb{E}\left\{ \int_0^{\mathcal{T}_k}e^{-\beta_k t} dc_k(t)   + e^{-\beta_k \mathcal{T}_k}  \left[b_k(\mathcal{T}_k) + P(\mathcal{T}_k) X_k(\mathcal{T}_k) \right] \right\} \\
& \mbox{ subject to } \eqref{eq:Xnu2}, \eqref{eq:ACP}, (\mbox{C}1), (\mbox{C}2).
\end{aligned}
\end{equation}
Our goal is to understand the collective behavior of the miners.
Without loss of generality, assume that the miners are {\em interchangeable},
so we can drop the the subscript `$k$' in \eqref{eq:obj11} to get
the objective of a {\em typical} miner:
 \begin{align}
\label{eq:obj121}
U(x):= & \sup_{\{(\nu(t), b(t))\}} J(\nu,b):= \mathbb{E}\left\{ \int_0^{\mathcal{T}}e^{-\beta t} dc(t) + e^{-\beta \mathcal{T}} \left[ b(\mathcal{T}) + P(\mathcal{T}) X(\mathcal{T}) \right]\right\} \\
& \mbox{ subject to } X'(t) = \nu(t) + \frac{N'(t)}{N(t) - Z(t)} X(t), \, X(0) = x, \tag{C0} \\
& \qquad \qquad \quad \, dc(t) =  rb(t)dt-db(t) - P(t) \nu(t) dt - N'(t) \, L\left(\frac{\nu(t)}{N'(t)} \right) dt, \tag{C1} \\
& \qquad \qquad \quad \,  b(0) = 0, \, b(t) \ge 0 \mbox{ and } 0 \le X(t) \le N(t), \tag{C2} \\
& \qquad \qquad \quad \, P(t) = P(0) + \sigma B(t) - \eta(Z(t) - Z(0)), \tag{C3}
\end{align}
where (C0) is a repeat of the state dynamics in \eqref{eq:Xnu2}, and (C3) is the price dynamics in \eqref{eq:ACP}.
Compared to \cite{TY23b}, 
the volume constraint $\nu(t) \le \overline{\nu}$ for $\overline{\nu} > 0$ is removed;
instead the transaction cost $- N'(t) \, L\left(\frac{\nu(t)}{N'(t)} \right) dt$
is introduced in the budget constraint (C1).
As a result, the miner's strategy will no longer be a bang-bang control 
but depend on the market impact mechanism.

\quad Denote
\begin{equation}
\label{eq:Pbeta}
\widetilde{P}_\beta(t): = e^{-\beta t} \mathbb{E} P(t) = e^{-\beta t} \left[ P(0) - \eta (Z(t) - Z(0)) \right] \quad \mbox{and} \quad
\widetilde{P}(t): = \widetilde{P}_0(t).
\end{equation}
A simple argument (as in \cite{RS21, TY23b}) shows that 
the consumption-investment problem \eqref{eq:obj121} is separable:
\begin{equation*}
U(x) := \sup_{\{(\nu, b)\}} J( \nu,b) = \sup_{b} J_1(b) + 
\sup_{\nu} J_2(\nu),
\end{equation*}
where $J_1(b):=(r - \beta)\int_0^{\mathcal{T}} e^{-\beta t} b(t) dt$ and
\begin{equation*}
J_2(\nu):=\int_0^\mathcal{T} \left[-\widetilde{P}_\beta(t) \nu(t) - e^{-\beta t} N'(t) L \left( \frac{\nu(t)}{N'(t)}\right) \right] dt 
+ \widetilde{P}_\beta(\mathcal{T})X(\mathcal{T}).
\end{equation*}

\quad Here we are concerned with the risk-averse setting, corresponding to $\beta \ge r$.
Then $\sup_b J_1(b) = 0$ with the optimality binding at $b_*(t) = 0$ for all $t$. 
So the problem \eqref{eq:obj121} is reduced to:
\begin{equation}
\label{eq:45}
U(x) = \sup_{\nu} J_2(\nu) \quad \mbox{subject to (C0), (C2')},
\end{equation}
where (C2') is (C2) without the constraints on $b(\cdot)$. 
We argue by dynamic programming, and let
\begin{align*}
v(t,x):= & \sup_{\{\nu(s), s \ge t\}} \int_t^{\mathcal{T}}
\left[-\widetilde{P}_\beta(s) \nu(s) - e^{-\beta s} N'(s) L \left( \frac{\nu(s)}{N'(s)}\right) \right] dt 
+ \widetilde{P}_\beta(\mathcal{T})X(\mathcal{T}) \\
& \mbox{ subject to } X'(s) = \nu(s) + \frac{N'(s)}{N(s) - Z(s)} X(s), \, X(t) = x,  \\
& \qquad \qquad \quad \, 0 \le X(s) \le N(s),
\end{align*}
so $U(x) = v(0,x)$. 

\quad Let $Q: = \{(t,x): 0 \le t < T, \, 0<x<N(t)\}$.
Under suitable conditions on the model parameters,
$v$ solves (uniquely) the HJB equation:
\begin{equation*}
\left\{ \begin{array}{lcl}
\partial_t v  + \frac{x N'(t)}{N(t)-Z(t)} \partial_x v + \sup_{\nu} \left\{\nu ( \partial_x v - \widetilde{P}_{\beta}(t)) - e^{-\beta t} N'(t) L\left(\frac{\nu}{N'(t)} \right) \right\} = 0 \quad \mbox{in } Q , \\
v(T,x) = \widetilde{P}_\beta(T)x, \\
v(t,0) = 0, \,\, v(t, N(t)) = \widetilde{P}_\beta(t)N(t).
\end{array}\right.
\end{equation*}

By optimizing $\nu \to \nu ( \partial_x v - \widetilde{P}_{\beta}(t)) - e^{-\beta t} N'(t) L\left(\frac{\nu}{N'(t)} \right)$, 
we get
$\nu_* = N'(t) (L')^{-1} \left(e^{\beta t}  \partial_x v - \widetilde{P}(t)\right)$.
This yields the following nonlinear PDE:
\begin{equation}
\label{eq:HJB}
\left\{ \begin{array}{lcl}
\partial_t v +  \frac{x N'(t)}{N(t) - Z(t)} \partial_x v + e^{-\beta t} N'(t) \bigg\{(e^{\beta t}  \partial_x v - \widetilde{P}(t)) (L')^{-1}(e^{\beta t}  \partial_x v - \widetilde{P}(t)) \\
\qquad \qquad  \qquad \qquad \qquad \qquad \qquad \qquad \, - L\left( (L')^{-1}(e^{\beta t}  \partial_x v - \widetilde{P}(t))\right) \bigg\} = 0 \quad \mbox{in }Q, \\
v(T,x) = \widetilde{P}_\beta(T)x, \\
v(t,0) = 0, \,\, v(t, N(t)) = \widetilde{P}_\beta(t)N(t).
\end{array}\right.
\end{equation}

{\em Example}: 
The most important case is the Almgren-Chriss model, with the quadratic cost $L(x) = \rho x^2$.
The PDE \eqref{eq:HJB} specializes to
\begin{equation}
\label{eq:HJBsqrt}
\left\{ \begin{array}{lcl}
\partial_t v + \frac{x N'(t)}{N(t) - Z(t)} \partial_x v + \frac{e^{\beta t} N'(t)}{4 \rho}(\partial_x v - \widetilde{P}_\beta(t))^2 = 0 \quad \mbox{in }Q, \\
v(T,x) = \widetilde{P}_\beta(T)x, \\
v(t,0) = 0, \,\, v(t, N(t)) = \widetilde{P}_\beta(t)N(t),
\end{array}\right.
\end{equation}
with the optimal strategy $\nu_*(t,x) = \frac{N'(t)}{2 \rho}\left(e^{\beta t}  \partial_x v(t,x) - \widetilde{P}(t)\right)$.

\smallskip
{\bf Mean field strategy}.
Assume that the distribution of the miners by their staked coins
are approximated by $m_0(x) dx$ at time $t =0$.
The goal is to find an equilibrium trading strategy
 $\nu^{\tiny \mbox{eq}}(t \,|\, m_0)$, or simply $\nu^{\tiny \mbox{eq}}(t)$,
which can be viewed as the aggregated trading strategy among all the miners.

\quad Now let's describe the mean field model.
\begin{tcolorbox}
\textbf{Mean field problem $(\star)$}
\begin{enumerate}[itemsep = 3 pt]
\item
Since there are $K$ miners, by \eqref{eq:Z}, the equilibrium of unstaked tokens is:
\begin{equation}
\label{eq:Zeqbis}
Z^{\tiny \mbox{eq}}(t) = Z(0) - K \int_0^t \nu^{\tiny \mbox{eq}}(s) ds,
\end{equation}
\item
Given $Z^{\tiny \mbox{eq}}(\cdot)$,
the miner's optimal strategy is
\begin{equation}
\label{eq:nueqbis}
\nu^{\tiny \mbox{eq}}_*(t,x) = N'(t) (L')^{-1} \left(e^{\beta t}  \partial_x v(t,x) - \widetilde{P}(t)\right),
\end{equation}
where $v$ is the solution to \eqref{eq:HJBsqrt} with $Z(t) = Z^{\tiny \mbox{eq}}(t)$.
\item
The feedback control of a typical miner is 
$X'(t) = \nu_*^{\tiny \mbox{eq}}(t,X(t))+ \frac{N'(t)}{N(t) - Z^{\tiny \mbox{eq}}(t)} X(t)$.
Hence, the density of the miners by their staked coins  solves the continuity equation:
\begin{equation}
\label{eq:FP}
\partial_t m + \partial_x \left( \left( \nu_*^{\tiny \mbox{eq}}(t,x) + \frac{x N'(t)}{N(t) - Z^{\tiny \mbox{eq}}(t) } \right) m  \right) = 0, \quad m(0,x) = m_0(x).
\end{equation}
\item
The equilibrium trading strategy $\nu^{\tiny \mbox{eq}}(t)$ satisfies the fixed point equation:
\begin{equation}
\label{eq:fixedpt}
\int \nu_*^{\tiny \mbox{eq}}(t,x) m(t,x) dx = \nu^{\tiny \mbox{eq}}(t).
\end{equation}
\end{enumerate}
\end{tcolorbox}

\quad Let's assume that the mean field problem $(\star)$ is well-posed. 
We are interested in the following two statistics (or features):
\begin{itemize}[itemsep = 3 pt]
\item
$Z^{\tiny \mbox{eq}}(t)$ is the number of unstaked tokens at time $t$, 
or $\frac{N(t)-Z^{\tiny \mbox{eq}}(t)}{N(t)}$ is the staking ratio
that often reflects the adoption of a token,
or activities within the token's ecosystem.
\item
$m(t,\cdot)$ is the distribution of staked coins of the miner population at time $t$.
A more-concentrated staking profile $m(t,\cdot)$ implies centralization in blockchain governance;
while a more spread-out $m(t,\cdot)$ implies decentralization of the blockchain protocol.
\end{itemize}
The interesting (and important) problem is to understand the long-time behavior of the above quantities,
i.e., as $t \to \infty$.
It is known that the PoS protocol alone does not lead to decentralization \cite{RS21, Tang22}.
Our goal is to explore whether allowing trading in the PoS protocol can yield decentralization, 
leveraging the market power.

\quad Hence, it is indispensable to first establish the well-posedness of the mean field problem $(\star)$ for {\em all time}.
As pointed out in \cite{Tangsuv}, this is equivalent to solving a first-order mean field game with nonlocal and non-separable Hamiltonian,
which is very challenging.
In the next section,
we take the first step to prove the well-posedness of this problem, at least for some (small) finite time.

\section{Well-posedness of the mean field problem}
\label{sc3}

\quad In this section, we prove the well-posedness of the mean field problem $(\star)$ for short time $T_0$, via a fixed-point approach.

\quad Let's recast the mean field problem. Fixing $M >0$, we look for an aggregated trading strategy in the set
\begin{equation*}
\mathcal{X}_M: = \{\nu \in \mathcal{C}[0,T_0]: ||\nu||_\infty \le M\}.
\end{equation*}
Write 
\begin{equation}
\label{eq:ZP}
Z_\nu(t):= Z(0) - K \int_0^t \nu(s)ds \quad \mbox{and} \quad
\widetilde{P}_{\beta, \nu}(t):=e^{-\beta t}\left(P(0) + \eta K \int_0^t \nu(s)ds \right).
\end{equation}
Choose $\overline{N} > \sup_{0 \le t \le T_0} N(t)$.
Formally, let $v_\nu(t,x)$ solve the HJB equation \eqref{eq:HJB}
with a fixed boundary\footnote{Here is a subtlety: compared to the original equation \eqref{eq:HJB},
the boundary conditions are removed in \eqref{eq:HJBbis}.
In fact, we will show in Lemma \ref{prop:inside} that under suitable condition on $m_0$,
$m(t,\cdot)$ or $X(t)$ remains (uniformly) in the interior of $Q_T$.
As a result, the boundary conditions and choice of $\overline{N}$ become irrelevant.}:
\begin{equation}
\label{eq:HJBbis}
\left\{ \begin{array}{lcl}
\partial_t v +  a_\nu(t,x) \partial_x v + g(t) f\left(e^{\beta t} (\partial_x v -\widetilde{P}_{\beta, \nu}(t))\right)= 0 \quad \mbox{in } [0,T_0) \times (0, \overline{N}), \\
v(T_0,x) = \widetilde{P}_{\beta, \nu}(T_0)x, 
\end{array}\right.
\end{equation}
where
\begin{equation}
\label{eq:fg}
a_\nu(t,x):= \frac{x N'(t)}{N(t) - Z_\nu(t)}, \quad
g(t):= e^{-\beta t} N'(t), \quad
f(p):= p (L')^{-1}(p) - L((L')^{-1}(p)).
\end{equation}
The corresponding feedback control is
$\mu_\nu(t,x):=N'(t) (L')^{-1} \left(e^{\beta t} (\partial_x v_\nu -\widetilde{P}_{\beta, \nu}(t))\right)$.
Let $m_\nu(t,x)$ solve the continuity equation \eqref{eq:FP}:
\begin{equation}
\label{eq:FPbis}
\partial_t m + \partial_x(b_\nu(t,x) m) = 0, \quad m(0,x) = m_0(x),
\end{equation}
where
\begin{equation*}
b_\nu(t,x):= \mu_\nu(t,x) + a_\nu(t,x).
\end{equation*}
Finally, define the fixed-point map
\begin{equation}
\label{eq:fixedmap}
(\mathcal{T}\nu)(t):= \int_0^{\overline{N}} \mu_\nu(t,x) m_\nu(t,x) dx.
\end{equation}
A fixed point of $\mathcal{T}$ on $\mathcal{X}_M$ gives a solution to the mean field problem $(\star)$.

\quad The main result is stated as follows. 
\begin{theorem}
\label{thm:wellp}
Assume that
\begin{enumerate}[itemsep = 3 pt]
\item
$N \in \mathcal{C}^2([0,T_0])$, with $N(0) > 0$ and $N'(t) \ge 0$ for all $t$. 
\item
$\delta_0:= N(0) - Z(0) > 0$.
\item
$L \in  \mathcal{C}^2(\mathbb{R})$, with $(L')^{-1}$ being Lipschitz.
\item
$m_0 \in \mathcal{C}([0, \overline{N}])$, with $m_0 \ge 0$, $\int_0^{\overline{N}} m_0(x) dx = 1$ and 
$\operatorname{supp} m_0 \subset [\varepsilon_0, N(0) - \varepsilon_0]$ for some $\varepsilon_0 > 0$.
\end{enumerate}
For $M$ sufficiently large and $T_0$ sufficiently small, 
the map $\mathcal{T}$ is given by
\begin{equation}
\label{eq:keyf}
(\mathcal{T}\nu)(t) =  N'(t) (L')^{-1} \left(e^{\beta t} \left(\widetilde{P}_{\beta, \nu}(T_0)e^{\int_t^{T_0} a_\nu(s) ds} -\widetilde{P}_{\beta, \nu}(t)\right)\right),
\end{equation}
and it has a unique fixed point $\nu^{\tiny \mbox{eq}} \in \mathcal{X}_M$.
The fixed problem $nu^{\tiny \mbox{eq}}$ is the equilibrium trading strategy of the mean field problem $(\star)$,
and the corresponding $v(t,x)$ and $m(t,x)$ are given by \eqref{eq:vformula} and \eqref{eq:mfor} respectively,
with $\nu =\nu^{\tiny \mbox{eq}}$.
\end{theorem}

\quad The key takeaway of the theorem is that under the linear terminal condition $v(T_0,x) = \widetilde{P}_{\beta, \nu}(T_0)x$,
the HJB equation \eqref{eq:HJBbis} is solvable. 
In particular, the miner's optimal strategy \eqref{eq:nueqbis} is:
\begin{equation*}
\nu^{\tiny \mbox{eq}}_*(t,x) =  N'(t) (L')^{-1} \left(e^{\beta t} \left(\widetilde{P}_{\beta, \nu}(T_0)e^{\int_t^{T_0} a_\nu(s) ds} -\widetilde{P}_{\beta, \nu}(t)\right)\right), \quad \mbox{with } \nu = \nu^{\tiny \mbox{eq}},
\end{equation*}
which is independent of the state variable $x$.
That is, the miner's trading decision is independent of her staked tokens.
So the  problem is to find $\nu(\cdot)$ such that
\begin{equation}
\label{eq:finalp}
N'(t) (L')^{-1} \left(e^{\beta t} \left(\widetilde{P}_{\beta, \nu}(T_0)e^{\int_t^{T_0} a_\nu(s) ds} -\widetilde{P}_{\beta, \nu}(t)\right)\right) = \nu(t).
\end{equation}
The proof of Theorem \ref{thm:wellp} is given in the next subsections.

\subsection{Solving the equations for fixed $\nu$}

We prove that the fixed-point map \eqref{eq:fixedmap} is well-defined
by solving the HJB equation \eqref{eq:HJBbis} and the continuity equation \eqref{eq:FPbis}.

\quad We start by showing that the denominator $N(t) - Z_\nu(t)$ in the coefficients $a_\nu(t,x)$, $b_\nu(t,x)$ 
stays away from zero for some short time $T_0$.
\begin{lemma}
\label{lem:zero}
Let $\nu \in \mathcal{X}_M$ and $\delta_0:= N(0) - Z(0)>0$.
Assume that $N$ is nondecreasing.
For $T_0 \le \frac{\delta_0}{2KM}$, we have:
\begin{equation*}
\inf_{0 \le t \le T_0} \{N(t) -Z_\nu(t)\} \ge \frac{\delta_0}{2}.
\end{equation*}
\end{lemma}
\begin{proof}
Note that
\begin{equation*}
\begin{aligned}
N(t) - Z_\nu(t) &= N(t) - Z(0) + K \int_0^t \nu(s) ds \\
& \ge \delta_0 - KMt,
\end{aligned}
\end{equation*}
where the inequality follows from the fact that $N$ is nondecreasing
and $||\nu||_\infty \le M$.
Thus, $N(t) - Z_\nu(t) \ge \delta_0/2$ for any $t \le \frac{\delta_0}{2KM}$.
 \end{proof}

\quad Next we solve explicitly the HJB equation \eqref{eq:HJBbis}.
\begin{proposition}
\label{prop:HJB}
Let the assumptions in Lemma \ref{lem:zero} hold,
and assume that $N\in \mathcal{C}^2([0,T_0])$ 
and $(L')^{-1}$ is Lipschitz.
We have:
\begin{equation}
\label{eq:vformula}
\begin{aligned}
v_\nu(t,x) &= \widetilde{P}_{\beta, \nu}(T_0)e^{\int_t^{T_0} a_\nu(s) ds}x + \int_t^T  g(s) f\left(e^{\beta s} \left(\widetilde{P}_{\beta, \nu}(T_0) e^{\int_s^{T_0} a_\nu(r) dr} -\widetilde{P}_{\beta, \nu}(s)\right)\right) ds,
\end{aligned}
\end{equation}
where $f$ and $g$ are defined by \eqref{eq:fg}.
Consequently, the feedback control is
\begin{equation}
\label{eq:munu}
\mu_\nu(t,x) = \mu_\nu(t):=N'(t) (L')^{-1} \left(e^{\beta t} \left(\widetilde{P}_{\beta, \nu}(T_0)e^{\int_t^{T_0} a_\nu(s) ds} -\widetilde{P}_{\beta, \nu}(t)\right)\right).
\end{equation}
\end{proposition}
\begin{proof}
Formally, differentiating \eqref{eq:HJBbis} with respect to $x$ yields
the quasi-linear equation for $q_\nu:= \partial_x v_\nu$:
\begin{equation}
\label{eq:quasi}
\partial_t q + c_\nu(t,x,q) \partial_x q + a_\nu(t)q =0, \quad q(T_0,x) = \widetilde{P}_{\beta, \nu}(T_0),
\end{equation}
where $a_\nu(t):= \partial_x a_\nu(t,x) =\frac{N'(t)}{N(t) - Z(t)}$ and 
\begin{equation*}
\begin{aligned}
c_\nu(t,x,q):&= a_\nu(t,x) + N'(t)f'\left(e^{\beta t} (q -\widetilde{P}_{\beta, \nu}(t))\right) \\
& = a_\nu(t,x) +  N'(t)(L')^{-1}\left(e^{\beta t} (q -\widetilde{P}_{\beta, \nu}(t))\right).
\end{aligned}
\end{equation*}
By Lemma \ref{lem:zero}, the coefficients $c_\nu(t,x,q)$ and $a_\nu(t)$ are well-defined,
continuous in $(t,x,q)$ and Lipschitz in $(x,q)$.
Thus, the characteristic system (see \cite[Chapter 3]{Evans})
\begin{equation}
\label{eq:char}
\left\{
\begin{aligned}
x'(t)&=c_\nu(t,x(t),q(t)),\\
q'(t)&=-\alpha_\nu(t)q(t),
\end{aligned}
\right.
\qquad
x(0)=x,
\quad q(T_0)=\widetilde{P}_{\beta, \nu}(T_0),
\end{equation}
has a unique local solution. 
Since $N$ is smooth, we get $|a_\nu(t)| \le \frac{2 ||N'||_\infty}{\delta}$.
By Gr\"{o}nwall's inequality,
\begin{equation*}
||q_{\nu}||_\infty \le \widetilde{P}_{\beta, \nu}(T_0) e^{\frac{2 ||N'||_\infty T_0}{\delta}}
\le  (|P(0)| + \eta K M T_0) e^{(\frac{2 ||N'||_\infty}{\delta} - \beta )T_0}.
\end{equation*}
Further by the fact that $(L')^{-1}$ is Lipschitz (so it has at most linear growth),
we get:
\begin{equation*}
|c_\nu(t,x,q)| \le \frac{2 ||N'||_\infty}{\delta} |x| + C(1 + |q|),
\end{equation*}
for some $C > 0$ independent of $\nu$.
Again Gr\"{o}nwall's inequality guarantees that $x_\nu(\cdot)$ does not blow up.
As a result, the solution is global:
\begin{equation*}
q_\nu(t,x) = \widetilde{P}_{\beta, \nu}(T_0) e^{\int_t^{T_0} a_\nu(s) ds},
\end{equation*}
and 
$v_\nu(t,x) = \int_t^T  g(s) f\left(e^{\beta s} (q_\nu(s,0) -\widetilde{P}_{\beta, \nu}(s))\right) ds+ \int_0^x q_\nu(t,y)dy$,
which proves \eqref{eq:vformula}.
\end{proof}

\quad Observe a subtlety: compared to the original equation \eqref{eq:HJB},
the boundary conditions are removed in \eqref{eq:HJBbis}.
Here we argue that the support of $m(t,\cdot)$ remains in the interior for short time,
so the choice of $\overline{N}$ becomes irrelevant. 

\begin{lemma}
\label{prop:inside}
Let the assumptions in Proposition \ref{prop:HJB} hold, and 
assume that $\operatorname{supp} m_0 \subset [\varepsilon_0, N(0) - \varepsilon_0]$ for some $\varepsilon_0 > 0$.
There is  $T_0 > 0$ such that for each $\nu \in \mathcal{X}_M$, 
the characteristic $x_\nu(\cdot)$ defined by \eqref{eq:char} starting from $x \in \operatorname{supp} m_0$  satisfy
\begin{equation*}
x_\nu(t) \in \left[ \frac{\varepsilon_0}{2}, N(t)-\frac{\varepsilon_0}{2}\right] \quad \mbox{for all } t \in [0,T_0].
\end{equation*}
\end{lemma}
\begin{proof}
It is easy to deduce from the proof of Proposition \ref{prop:HJB} that
\begin{equation*}
|c(t,x,q_\nu(t,x))| \le C \quad \mbox{for some } C > 0 \mbox{ independent of } \nu.
\end{equation*}
So $|x'_\nu(t)| \le C$, and hence $|x_\nu(t) - x| \le Ct$.
Choosing $T_0 \le \frac{\varepsilon_0}{2C}$ yields
\begin{equation*}
x_\nu(t) \ge x - Ct \ge \varepsilon_0- \frac{\varepsilon_0}{2} = \frac{\varepsilon_0}{2},
\end{equation*}
because $x \ge \varepsilon_0$.
For the upper bound, we have:
\begin{equation*}
N(t) - x_\nu(t) \ge N(0) - x - Ct  \ge \varepsilon_0 - \frac{\varepsilon_0}{2} = \frac{\varepsilon_0}{2},
\end{equation*}
because $N(\cdot)$ is nondecreasing and $x \le N(0) - \varepsilon_0$.
Thus, $x_\nu(t)$ remains in the interior of $Q = \{(t,x): 0 \le t < T_0, \, 0<x<N(t)\}$.
\end{proof}

\quad The following proposition solves the continuity equation \eqref{eq:FPbis}.
\begin{proposition}
\label{prop:FP}
Let the assumptions in Lemma \ref{prop:inside} hold.
For each $\nu \in \mathcal{X}_M$, 
the continuity equation \eqref{eq:FPbis} has a unique classical solution $m \in \mathcal{C}([0,T_0]; L^1(0, \overline{N}))$ 
given by the push-forward of $m_0$ by the flow $b_\nu$:
\begin{equation}
\label{eq:mfor}
m(t,x) = m_0\left( e^{-\int_0^t a_\nu(s) ds}x -  \int_0^t e^{-\int_0^s a_\nu(r) dr} \mu_\nu(s)ds\right) e^{-\int_0^t a_\nu(s)ds}.
\end{equation}
In particular,
$m(t, \cdot) \ge 0$ and $\int_{\frac{\varepsilon_0}{2}}^{N(0) - \frac{\varepsilon_0}{2}} m(t,x) dx = 1$.
\end{proposition}
\begin{proof}
Recall the expression of $\mu_\nu(t,x)$ from \eqref{eq:munu}.
Because $N \in \mathcal{C}^2([0,T_0])$ and $(L')^{-1} \in \mathcal{C}^1(\mathbb{R})$, 
we have $\mu_\nu(t,x) \in \mathcal{C}^1([0,T_0] \times [0, \overline{N}])$,
and hence,
\begin{equation*}
b_\nu(t,x) = a_\nu(t)x + \mu_\nu(t) \in \mathcal{C}^1([0,T_0] \times [0, \overline{N}])
\end{equation*}
As a result, the ordinary differential equation
\begin{equation*}
X'(t) = b_\nu(t, X(t)), \quad X(0) = x,
\end{equation*}
has a unique solution on $[0,T_0]$:
\begin{equation*}
\Phi_\nu(t;x) = e^{\int_0^t a_\nu(s) ds}x + e^{\int_0^t a_\nu(s) ds} \int_0^t e^{-\int_0^s a_\nu(r) dr} \mu_\nu(s)ds.
\end{equation*}
The unique distributional solution to \eqref{eq:FPbis} is the push-forward 
$m(t,\cdot) = \Phi_\nu(t; \cdot)_{\#}m_0$, or in density form:
\begin{equation*}
m(t,x) = m_0(\Phi^{-1}_\nu(t;x))\exp\left(-\int_0^t \partial_x b_\nu(s, \Phi(s;\Phi^{-1}_\nu(t;x)))ds\right),
\end{equation*}
which implies \eqref{eq:mfor}.
The fact that $\operatorname{supp} m(t, \cdot) \subset \left[ \frac{\varepsilon_0}{2}, N(t)-\frac{\varepsilon_0}{2}\right]$ follows from Lemma \ref{prop:inside}.
\end{proof}

\subsection{Stability estimates}
We quantify the stability of the equations \eqref{eq:HJBbis} and \eqref{eq:FPbis} with respect to the control $\nu$.

\quad We start with the Lipschitz estimates of the coefficients.
\begin{lemma}
\label{lem:coef}
Let the assumptions in Proposition \ref{prop:HJB} hold.
For $\nu_1, \nu_2 \in \mathcal{X}_M$, we have:
\begin{align}
&||Z_{\nu_1} - Z_{\nu_2}||_\infty \le K T_0 ||\nu_1 - \nu_2||_\infty, \label{Z}\\
&||\widetilde{P}_{\beta, \nu_1} - \widetilde{P}_{\beta, \nu_2}||_\infty \le \eta K T_0 ||\nu_1 - \nu_2||_\infty, \label{P} \\
&||a_{\nu_1} - a_{\nu_2} ||_\infty \le \frac{4 \overline{N} ||N'||_\infty K T_0}{\delta_0^2} ||\nu_1 - \nu_2||_\infty. \label{a}
\end{align}
\end{lemma}
\begin{proof}
The bounds \eqref{Z}--\eqref{P} follow immediately from \eqref{eq:ZP}.
For \eqref{a}, we get:
\begin{equation*}
\begin{aligned}
|a_{\nu_1}(t,x) - a_{\nu_2}(t,x)|
& = xN'(t) \frac{|Z_{\nu_2}(t) - Z_{\nu_1}(t)|}{(N(t) - Z_{\nu_1}(t))(N(t) - Z_{\nu_2}(t))} \\
& \le \frac{4 \overline{N} ||N'||_\infty K T_0}{\delta_0^2} ||\nu_1 - \nu_2||_\infty,
\end{aligned}
\end{equation*}
where we use Lemma \ref{lem:zero} to lower bound $N(t) - Z_\nu(t)$,
and \eqref{Z} to upper bound $|Z_{\nu_2}(t) - Z_{\nu_1}(t)|$.
\end{proof}

\quad Next we study the stability of the feedback control $\mu_\nu(t,x)$ with respect to $\nu$.
\begin{proposition}
\label{prop:HJBst}
Let the assumptions in Proposition \ref{prop:HJB} hold. 
For $T_0$ sufficiently small, 
we have for $\nu_1, \nu_2 \in \mathcal{X}_M$,
\begin{equation}
\label{eq:mu}
||\mu_{\nu_1} - \mu_{\nu_2}||_\infty \le C T_0   ||\nu_1 - \nu_2||_\infty,
\end{equation}
for some $C > 0$ independent of $T_0$.
\end{proposition}
\begin{proof}
Because $N \in \mathcal{C}^2([0,T_0])$ and $(L')^{-1}$ is Lipschitz,
we have:
\begin{equation}
\label{eq:muest}
\begin{aligned}
&|\mu_{\nu_1}(t,x) - \mu_{\nu_2}(t,x)| \\
&\quad \le C||N'||_\infty \left(\left|\widetilde{P}_{\beta, \nu_1}(T_0)e^{\int_t^{T_0} a_{\nu_1}(s) ds} - \widetilde{P}_{\beta, \nu_2}(T_0)e^{\int_t^{T_0} a_{\nu_2}(s) ds}\right| + |\widetilde{P}_{\beta, \nu_1}(t) - \widetilde{P}_{\beta, \nu_2}(t)|\right) \\
& \quad \le C||N'||_\infty \left((e^{\int_t^{T_0} a_{\nu_1}(s) ds} + 1) |\widetilde{P}_{\beta, \nu_1}(t) - \widetilde{P}_{\beta, \nu_2}(t)|
+ \widetilde{P}_{\beta, \nu_2}(T_0)\left|e^{\int_t^{T_0} a_{\nu_1}(s) ds}  - e^{\int_t^{T_0} a_{\nu_2}(s) ds} \right| \right)
\end{aligned}
\end{equation}
By Lemma \ref{lem:zero}, we get $||a_\nu||_\infty \le \frac{2||N'||_\infty}{\delta_0}$,
and hence 
\begin{equation}
\label{eq:exp}
e^{\int_t^{T_0} a_{\nu}(s) ds} \le e^{\frac{2||N'||_\infty T_0}{\delta_0}}.
\end{equation}
Combining \eqref{P}, \eqref{a}, \eqref{eq:muest} and \eqref{eq:exp} yields
\begin{equation*}
|\mu_{\nu_1}(t,x) - \mu_{\nu_2}(t,x)| \le C T_0(1 + e^{CT_0}+ e^{CT_0} T_0 + e^{CT_0} T_0^2) ||\nu_1 - \nu_2||_\infty,
\end{equation*}
which proves \eqref{eq:mu}.
\end{proof}

\subsection{Proof of Theorem \ref{thm:wellp}}

The proof is split into several steps.

\smallskip
\noindent
{\bf Step 1}. We prove that for $M$ sufficiently large and $T_0$ sufficiently small, $\mathcal{T}(\mathcal{X}_M) \subset \mathcal{X}_M$.
Because $N \in \mathcal{C}^2([0,T_0])$ and $(L')^{-1}$ is Lipschitz,
we have:
\begin{equation*}
||\mu_\nu||_\infty \le C(1+e^{CT_0})(1+\eta KM T_0),
\end{equation*}
for some $C > 0$ independent of $T_0$.
Thus,
$|(\mathcal{T}\nu)(t)| \le ||\mu_\nu||_\infty \le C(1+e^{CT_0})(1+\eta KM T_0)$.
Choosing $M \ge 2C(1+e)$,
and then $T_0 \le \min(\frac{1}{C}, \frac{1}{\eta KM})$ yields $||\mu_\nu||_\infty \le 2C(1+e) \le M$.

\smallskip
\noindent
{\bf Step 2}. Note that $|(\mathcal{T}\nu_1)(t) - (\mathcal{T}\nu_2)(t)|  \le ||\mu_{\nu_1} - \mu_{\nu_2}||_\infty$.
By Proposition \ref{prop:HJBst}, 
we have for $T_0$ sufficiently small, 
\begin{equation*}
||\mathcal{T}\nu_1 - \mathcal{T}\nu_2||_\infty \le  C'T_0 ||\nu_1 - \nu_2||_\infty \quad \mbox{for all } \nu_1, \nu_2 \in \mathcal{X}_M.
\end{equation*}

\smallskip
\noindent
{\bf Step 3}. By the previous two steps, 
$\mathcal{T}$ is a strict contraction on the complete space $\mathcal{X}_M$, 
provided that $T_0 \le \frac{1}{C'}$.
By Banach fixed-point theorem, there is a unique $\nu^{\tiny \mbox{eq}} \in \mathcal{X}_M$ such that $\mathcal{T}\nu^{\tiny \mbox{eq}} = \nu^{\tiny \mbox{eq}}$.
The formulas for $v(t,x)$ and $m(t,x)$ follow from Propositions \ref{prop:HJB} and \ref{prop:FP}.

\section{Numerical experiments}
\label{sc4}

\quad In this section, we provide numerical experiments to solve the fixed-point problem \eqref{eq:finalp},
and hence, the mean field problem $(\star)$.
Our focus is on the evolution of staking ratio $\frac{N(t) - Z^{\tiny \mbox{eq}}(t)}{N(t)}$
and the staking profile $m(t,\cdot)$ over time,
thereby providing support to the claims made in the introduction.

\quad To start, we set:
\begin{itemize}[itemsep = 3 pt]
\item
$N(t) = 100 + 5t$ (linear token supply);
\item
$Z(0) = 50$ (initial staking ratio is $50\%$);
\item
$m_0 \sim \mathcal{N}(5, 1)$;
\item
$\beta = 0.05$, $\eta = 0.15$, $P(0) = 1$, $K = 10$;
\item
$L(x) = x^2$ (quadratic transaction cost).
\end{itemize}
We also take $M = 500$ and $T_0 = 2$.
For each $k$, 
validator $k$'s share is $\frac{X_k(t)}{N(t) - Z(t)}$,
so the {\em Herfindahl–Hirschman index} \cite{Hind} (HHI) is:
\begin{equation*}
H_K(t):= \sum_{k=1}^K \left(\frac{X_k(t)}{N(t) - Z^{\tiny \mbox{eq}}(t)}\right)^2.
\end{equation*}
In our mean-field formulation,
let
\begin{equation*}
M_1(t):= \int xm(t,x) dx \quad \mbox{and} \quad M_2(t):= \int x^2m(t,x) dx.
\end{equation*}
Accounting consistency requires $K M_1(t) = N(t) - Z^{\tiny \mbox{eq}}(t)$.
Thus,
$H_K(t) \approx \frac{K M_2(t)}{(K M_1(t))^2} = \frac{1}{K} \frac{M_2(t)}{M_1(t)^2}$.
To evaluate decentralization of the staking profile,
we consider the normalized HHI index:
\begin{equation}
\mathcal{H}(t):= \frac{M_2(t)}{M_1(t)^2}.
\end{equation}
We have $\mathcal{H}(t) \ge 1$, 
with equality if all validators have equal stakes.
The smaller the value of HHI, the more decentralized the profile is.
Figure \ref{fig:1} provides a numerical illustration to the mean field problem $(\star)$.
\begin{figure}[!htbp]  
    \centering
    \begin{subfigure}{0.48\textwidth}
        \centering
        \includegraphics[width=0.8\textwidth]{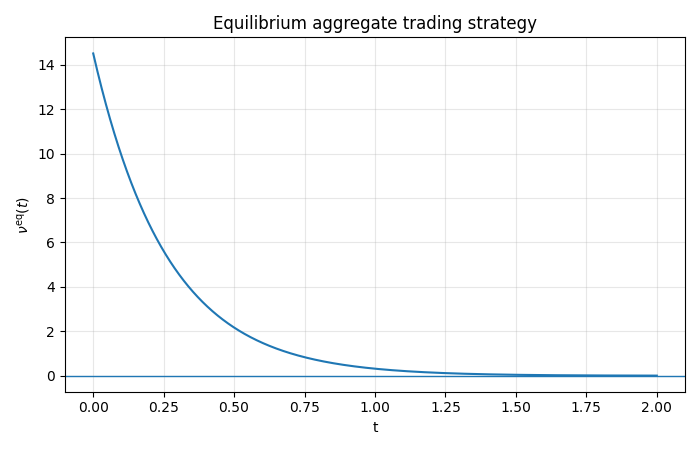}
        \caption{}
    \end{subfigure}
    \begin{subfigure}{0.48\textwidth}
        \centering
        \includegraphics[width=0.8\textwidth]{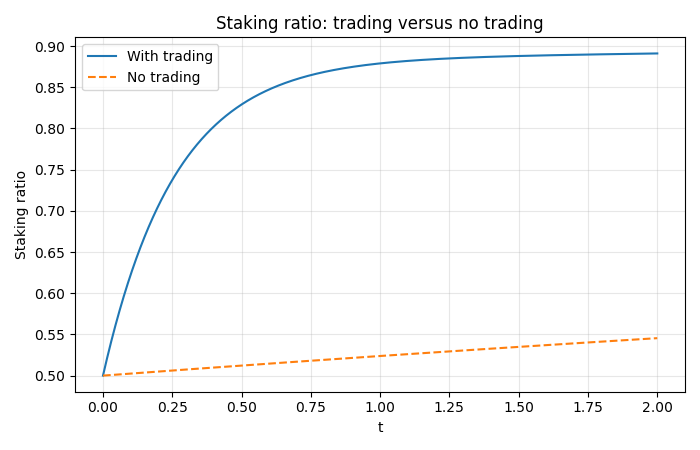}
        \caption{}
    \end{subfigure}
        \begin{subfigure}{0.48\textwidth}
        \centering
        \includegraphics[width=0.8\textwidth]{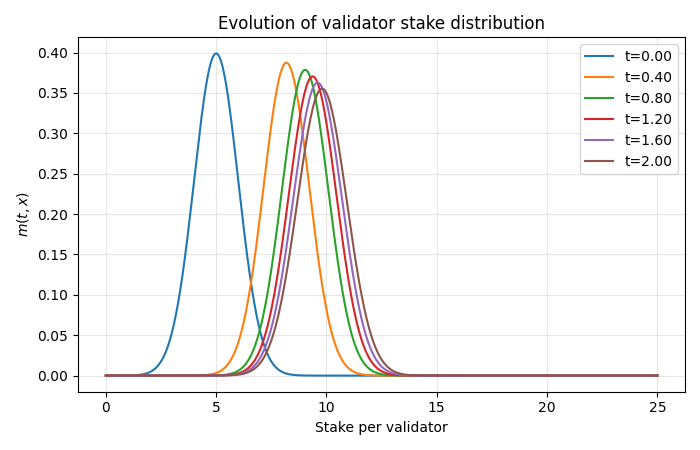}
        \caption{}
    \end{subfigure}
    \begin{subfigure}{0.48\textwidth}
        \centering
        \includegraphics[width=0.8\textwidth]{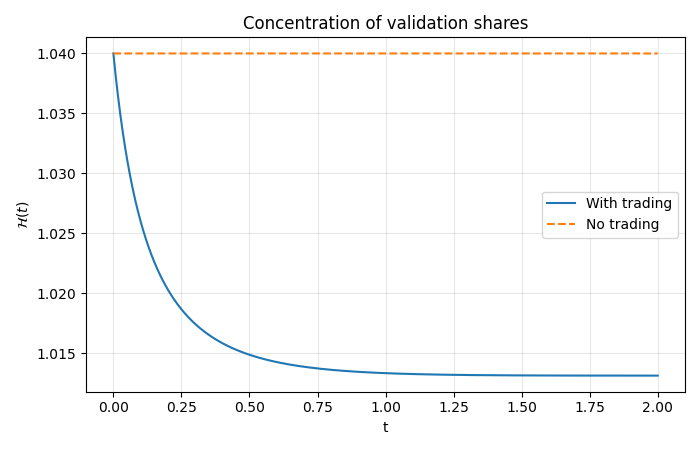}
        \caption{}
    \end{subfigure}
      \caption{Mean field problem $(\star)$ with $N(t) = 100 + 5t$, $L(x) = x^2$.}
     \label{fig:1}
\end{figure}

\quad Figure \ref{fig:1}(A) solves numerically the equilibrium trading strategy $\nu^{\tiny \mbox{eq}}(t)$.
Figure \ref{fig:1}(B) plots the staking ratio under $\nu^{\tiny \mbox{eq}}(t)$:
the ratio is increasing over time, and grows faster than that under no trading. 
This implies that trading can enhance the staking activity, 
and hence, the participation in the validation process of the PoS ecosystem.
Figure \ref{fig:1}(C) illustrates the evolution of the staking profile $m(t,\cdot)$,
and Figure \ref{fig:1}(D) computes the  HHI:
the staking profile decentralizes over time, 
demonstrating the power of market incentives in blockchain protocols. 
In fact, the decentralization of the staking profile can be justified by the following corollary.
\begin{corollary}
\label{coro:HHI}
Let the assumptions of Theorem \ref{thm:wellp} hold.
We have:
\begin{equation}
\label{eq:Hprim}
\mathcal{H}'(t) = \frac{2 \nu^{\tiny \mbox{eq}}(t)}{M_1(t)}(1 - \mathcal{H}(t)).
\end{equation}
Consequently,
if $\nu^{\tiny \mbox{eq}}(t) > 0$ (net purchases), then $\mathcal{H}'(t) \le 0$;
if $\nu^{\tiny \mbox{eq}}(t) < 0$ (net sales), then $\mathcal{H}'(t) \ge 0$;
and if $\nu^{\tiny \mbox{eq}}(t) \equiv 0$ (no trading), then $\mathcal{H}(t) \equiv \mathcal{H}(0)$.
\end{corollary}
\begin{proof}
Recall that under equilibrium,
each validator's stake is governed by
$X'(t) = a(t)X(t) + \nu^{\tiny \mbox{eq}}(t)$, with $a(t) = \frac{N'(t)}{N(t) - Z^{\tiny \mbox{eq}}(t)}$.
We compute
\begin{equation*}
M'_1(t) = a(t) M_1(t) + \nu^{\tiny \mbox{eq}}(t), \quad
M'_2(t) = 2a(t) M_2(t) + 2 \nu^{\tiny \mbox{eq}}(t) M_1(t).
\end{equation*}
As a result,
\begin{equation*}
\begin{aligned}
\mathcal{H}'(t) &= \frac{M'_2(t)}{M_1(t)^2} - \frac{2 M'_1(t) M_2(t)}{M_1(t)^3} \\
& = \frac{2a(t) M_2(t) + 2 \nu^{\tiny \mbox{eq}}(t) M_1(t)}{M_1(t)^2} - \frac{2 (a(t) M_1(t) + \nu^{\tiny \mbox{eq}}(t)) M_2(t)}{M_1(t)^3},
\end{aligned}
\end{equation*}
where the $a(t)$ terms cancel, and yields \eqref{eq:Hprim}.
\end{proof}

\quad As a consequence, net purchases enhance staking participation, 
and promote decentralization of the staking distribution,
this somewhat justifies \texttt{BitMine}'s ETH buy-and-stake strategy.

\medskip
{\bf Transaction costs}: We study the effect of the transaction cost on the mean field problem $(\star)$.
Fixing all the other parameters as above, we consider:
\begin{equation*}
L(x) = x^4 \quad \mbox{and} \quad L(x) = x^8.
\end{equation*}

\begin{figure}[!htbp]
    \centering
    \begin{subfigure}{\textwidth}
        \centering
        \includegraphics[width=0.85\textwidth]{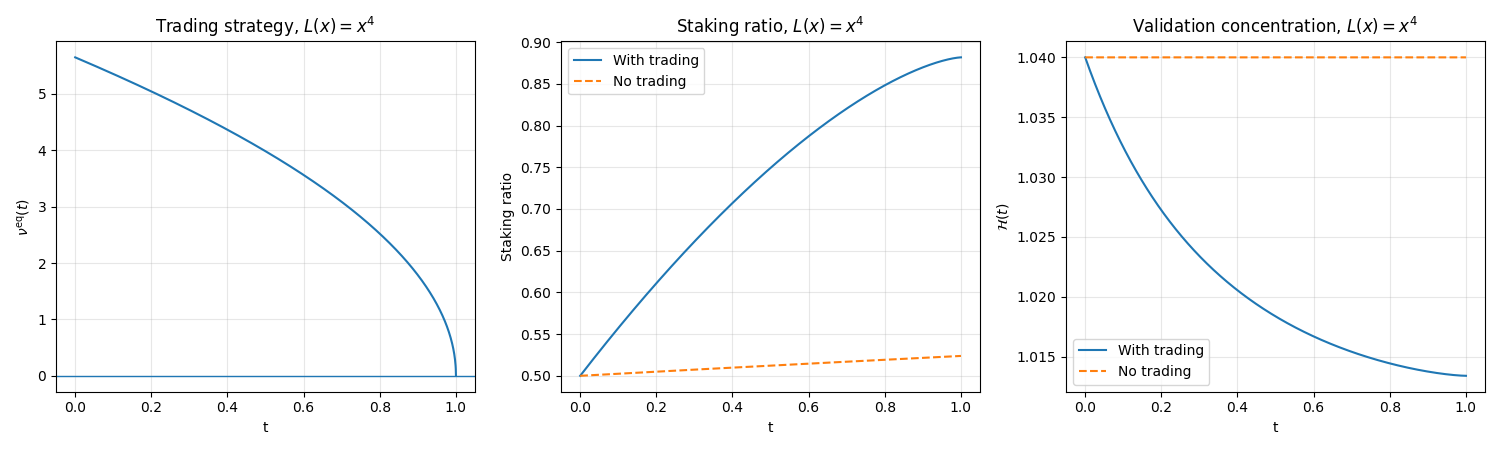}
    \end{subfigure}
     \begin{subfigure}{\textwidth}
        \centering
        \includegraphics[width=0.85\textwidth]{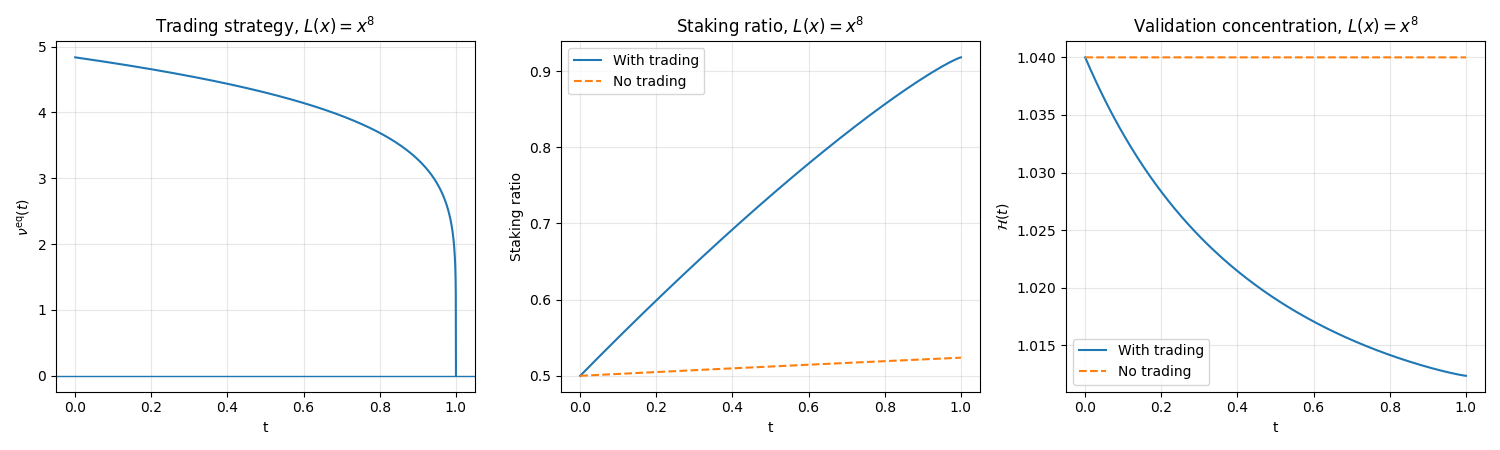}
    \end{subfigure}
     \caption{Mean field problem $(\star)$ with $N(t) = 100 + 5t$, $L(x) = x^4$ (top) and $L(x) = x^8$ (bottom).}
    \label{fig:2}
\end{figure}

Figure \ref{fig:2} plots the equilibrium trading strategy, the staking ratio, and the HHI 
for linear token supply and quartic/octic transaction costs.
Comparing with Figure \ref{fig:1}, we observe:
(1) Staking ratio increases over time,
and 
(2) The staking profile decentralizes over time.

\medskip
{\bf Token supply}: Here we consider different token supplies. 
Fixing all the other parameters as above, we set:
\begin{equation*}
N(t) = (100^{\frac{1}{\alpha}} + t)^{\alpha}, \quad \mbox{with } \alpha = 0.8, 1.0 \mbox{ and } 1.2.
\end{equation*}

\begin{figure}[!htbp]
    \centering
     \includegraphics[width=0.9\textwidth]{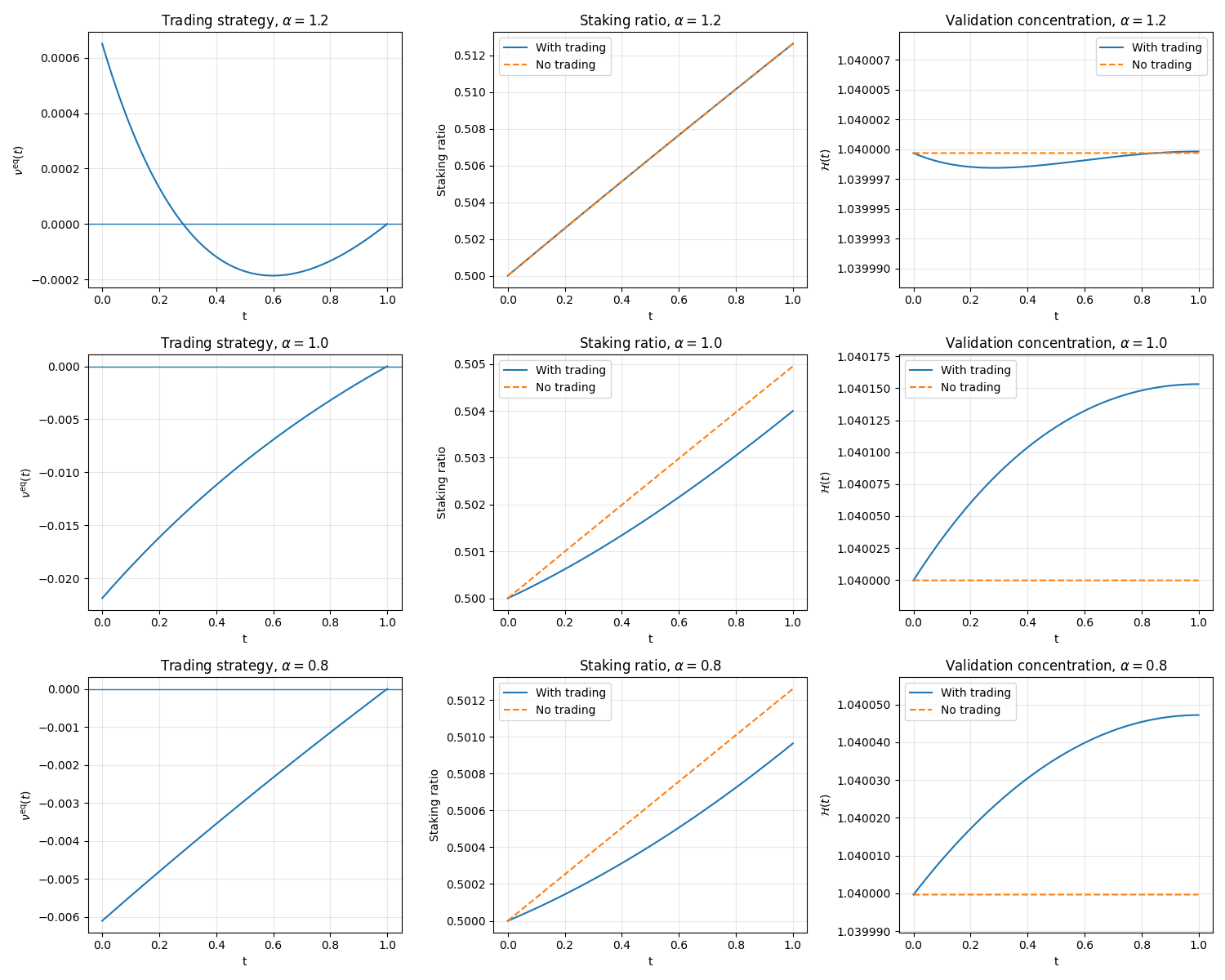}
     \caption{Mean field problem $(\star)$ with $N(t)=(100^{\frac{1}{\alpha}} + t)^{\alpha}$ with $\alpha = 1.2$ (top), $\alpha = 1.0$ (middle) and $\alpha = 0.8$ (bottom), $L(x) = x^2$.}
    \label{fig:3}
\end{figure}

Figure \ref{fig:3} plots the equilibrium trading strategy, the staking ratio, and the HHI 
for polynomial token supplies and quadratic transaction cost.

\quad We observe that the equilibrate trading rate $\nu(t)$ can be either decreasing (Figure \ref{fig:1}--\ref{fig:2}) or increasing (Figure \ref{fig:3}).
This can be explained by the following proposition, 
which shows that $\mu'(t)$ is determined by changing issuance $N''(t)$,
the discount/staking growth $\beta - a_\nu(t)$, and 
the market impact feedback $-\eta K \nu(t)$.

\begin{proposition}
Let $p_\nu(t):= P(0) + \eta K \int_0^T \nu(s)ds$,
$G(t):= e^{\int_t^{T_0} (a_\nu(s) -\beta) ds}$,
and $S(t):= p_\nu(T_0) G(t) - p_\nu(t)$.
Then we have:
\begin{equation}
\label{eq:nuprime}
\nu'(t) = N''(t) (L')^{-1}(S(t)) + \frac{N'(t)}{L''(L^{-1}(S(t)))}
\left((\beta - a_\nu(t)) p_\nu(T_0) G(t) - \eta K \nu(t)\right).
\end{equation}
\end{proposition}
\begin{proof}
By the fixed-point equation \eqref{eq:finalp}, we have:
\begin{equation}
\label{eq:nubibi}
\nu(t) = N'(t)  (L')^{-1}(S(t)).
\end{equation}
Differentiating $S(t)$ gives:
\begin{equation}
\label{eq:SS}
S'(t) = (\beta - a_\nu(t)) p_\nu(T_0) G(t) - \eta K \nu(t).
\end{equation}
Combining \eqref{eq:nubibi} and \eqref{eq:SS} yields \eqref{eq:nuprime}.
\end{proof}

\quad In Figure \ref{fig:1}--\ref{fig:2},
we have $N''(t) = 0$ and $a_\nu(0) = 0.1 > 0.05 = \beta$.
Thus, if $\nu(t) > 0$, then
\begin{equation*}
(\beta - a_\nu(t)) p_\nu(T_0) G(t) < 0 \quad \mbox{and}
\quad - \eta K \nu(t) < 0,
\end{equation*}
which implies $\nu'(t) < 0$.
Economically,
staking growth is sufficiently attractive relative to discounting, so miners buy strongly early. 
As the horizon approaches, there is less future staking reward to capture, so the optimal purchase rate declines.
In Figure \ref{fig:3}, we calculate
\begin{equation*}
a_\nu(0) =  \left\{ \begin{array}{rcl}
0.05171 & \mbox{for } \alpha = 1.2, \\ 
0.02000  & \mbox{for }  \alpha = 1.0, \\
0.00506 &\mbox{for }  \alpha = 0.8.
\end{array}\right.
\end{equation*}
For $\alpha = 1.2$, $a_\nu(t)$ is very close to $\beta = 0.05$,
so trading is negligible and the staking ratio nearly coincides with that under no trading.
For $\alpha = 1$ and $0.8$, 
we have $\beta > a_\nu(t)$ and $\nu(t) <0$ over much of the horizon.
Thus, 
\begin{equation*}
(\beta - a_\nu(t)) p_\nu(T_0) G(t) > 0 \quad \mbox{and}
\quad - \eta K \nu(t) > 0,
\end{equation*}
which implies $\nu'(t) > 0$.
In this case, discounting dominates staking growth,
so the incentive to retain tokens becomes weak.
By Corollary \ref{coro:HHI}, net sales lead to validator's concentration.

\section{Conclusion and discussions}
\label{sc5}

\quad In this paper, we studied the interaction between centralized trading activities and PoS blockchain ecosystems. 
By incorporating market impact from centralized exchanges into the PoS framework, we formulated a continuous-time mean field model describing the joint evolution of trading, staking, and token allocation among the miners. 
The resulting equilibrium problem couples an HJB equation with a transport equation via a nonlocal fixed-point condition.
We established the local well-posedness of the mean field system, and derived a semi-explicit characterization of the equilibrium trading strategy.
Numerical experiments demonstrate that  centralized trading activities may increase staking participation and promote decentralization of the staking distribution through market incentives. 

\quad The semi-explicit structure of the solution is largely due to the linear terminal condition imposed in the HJB equation. 
For a more general terminal condition, e.g., of the form $e^{-\beta T} h(x)$ with $h$ being interpreted as a utility function,
the solution to the HJB equation, and consequently the associated continuity equation, is no longer tractable. 
Moreover, the feedback control is time- and state-dependent.
Nevertheless, under suitable conditions on $h$, 
we can still establish the local well-posedness of the mean field system.
A very interesting question is to determine the maximal time horizon $T_0$ for which the mean field system remains well-posed.

\quad By a private communication with Paul Y. Zhang \cite{TZ+}, 
it is possible to approximate our mean field problem on bounded domains
by a mean field game on the whole space by introducing a penalty function near the boundary. 
Under the Almgren-Chriss model (i.e., quadratic transaction costs),
the argument in \cite{MM24} allows us to show that the mean field game is well-posed until $T_0 $ satisfying
$\left(\sup_{t \in [0,T_0]} e^{\beta t} N'(t) \right)^2 \eta T_0^2 \le 8\rho^2$.

\quad From the numerical perspective, the coupled HJB--continuity system also presents substantial challenges. In particular, numerical approximation of the continuity equation requires careful treatment in order to preserve positivity, mass conservation, and stability of the evolving density profile. Developing robust and efficient numerical schemes for the fully coupled mean field system remains an important direction for future research.

\bigskip
{\bf Acknowlegdment}: 
We thank David Yao, Qiang Du and Paul Yuming Zhang for helpful discussions at various stages of the project.
This research is supported by NSF CAREER Award DMS-2538791, NSF MFAI Award DMS-2602038
and a Columbia-CityU/HK collaborative project that is supported by InnoHK Initiative, The Government of the HKSAR and the AIFT Lab.

\bibliography{unique}
\bibliographystyle{abbrv}

\end{document}